\documentclass[journal]{IEEEtran}
\usepackage[numbers,sort&compress]{natbib}
\usepackage{amsfonts}
\usepackage[cmex10]{amsmath}
\usepackage{array}
\usepackage{mdwmath}
\usepackage{mdwtab}
\usepackage{eqparbox}
\usepackage[font=normalsize]{caption}
\usepackage{hyperref}
\usepackage{stfloats}
\usepackage{url}
\usepackage{amssymb}
\usepackage{float}
\usepackage{indentfirst}
\usepackage{makeidx}
\usepackage{tabularx}
\usepackage{cases}
\usepackage{multirow}
\usepackage{booktabs}
\usepackage{diagbox}
\usepackage{mathtools}
\usepackage{color}
\usepackage{stfloats}
\usepackage{lipsum}
\usepackage{amsfonts}
\usepackage{bm}
\usepackage{dsfont}
\usepackage{lipsum}
\usepackage{amsthm}
\usepackage{graphicx, epsfig, subfigure}
\usepackage[font=small]{caption}
\allowdisplaybreaks[4]

\makeatletter

\newif\if@restonecol
\makeatother

\usepackage[linesnumbered,ruled,vlined]{algorithm2e}
\usepackage{algpseudocode}
\usepackage{amsmath, bm}

\renewcommand{\arraystretch}{1.2}

\begin{document}
\title{Resource-aware Sub-network Allocation for Heterogeneous Decentralized Learning}
\author{Zhuojun~Tian, Ming~Xiao and Mikael~Skoglund, \textit{Fellow, IEEE}
\thanks{Z. Tian, M. Xiao and M. Skoglund are with the Department of Information Science and Engineering, School of Electrical Engineering and Computer Science, KTH Royal Institute of Technology, 10044 Stockholm, Sweden. (e-mail: zhuojun@kth.se; mingx@kth.se; skoglund@kth.se).} 
}
\maketitle

\begin{abstract}
System heterogeneity in computation, communication, and memory can substantially limit synchronous decentralized learning, where each training round is constrained by resource-limited nodes. Existing approaches provide limited support for jointly determining node-specific model capacities from deployment resources. We propose RASA-DL, a resource-aware sub-network allocation framework that treats the maximum sub-network width assigned to each node as a deployment variable. RASA-DL combines offline hardware profiling of computation and memory costs with deployment-aware profiling of learning utility under heterogeneous decentralized training. 
Using these profiles together with node- and link-level resource information, we formulate deployment as a bottleneck-latency minimization problem subject to per-node memory constraints and a required aggregate learning utility. Exploiting the monotone feasibility of the latency threshold, RASA-DL efficiently identifies the deployment through a binary-search-based allocation procedure. Experiments under controlled resource heterogeneity demonstrate favorable accuracy--latency tradeoffs. Under the Moderate and Severe regimes, the minimum-latency operating point improves accuracy over the largest globally feasible uniform deployment by 2.80 and 2.43 percentage points, respectively, at the same modeled round latency. The advantage is further amplified under strong Non-IID data heterogeneity.
\end{abstract}
\begin{IEEEkeywords}
Decentralized learning, heterogeneous edge systems, resource allocation, sub-network allocation, system heterogeneity.
\end{IEEEkeywords}


\section{Introduction}

Decentralized learning enables networked devices to collaboratively train machine learning models through peer-to-peer communication without relying on a centralized parameter server \cite{tian2026semantic}. By distributing both local computation and model exchange across participating nodes, it provides a flexible learning paradigm for edge and distributed networks while avoiding the communication bottleneck and single point of failure associated with centralized aggregation. In practical deployments, however, participating devices are rarely homogeneous. Their computation capabilities, communication bandwidths, and available memory can differ substantially, resulting in highly diverse per-round training costs. Such system heterogeneity is particularly problematic for synchronous decentralized learning, where the progress of each communication round is constrained by slow or resource-limited participants. Consequently, deploying an identical model architecture across all devices can lead to severe synchronization delays, inefficient utilization of resource-rich devices, or even memory infeasibility on constrained nodes.

Existing decentralized learning algorithms primarily focus on enabling efficient peer-to-peer optimization and model exchange \cite{tian2023distributed}. For example, DFedAvgM allows each node to perform multiple local updates and subsequently aggregate models received from neighboring nodes without centralized coordination \cite{sun2023decentralized}. Communication-efficient approaches further reduce network overhead through techniques such as sparse model exchange and adaptive peer communication \cite{tang2023gossipfl}. While these methods improve the efficiency and scalability of decentralized training, decentralization itself does not eliminate the impact of heterogeneous device resources. When all nodes are required to train models of comparable complexity, differences in computation speed, communication conditions, and memory capacity can still translate directly into system-level synchronization bottlenecks.

System heterogeneity has been extensively investigated in federated learning. Representative approaches accommodate heterogeneous devices by allowing variable amounts of local computation or by selectively involving efficient participants \cite{li2020fedprox,lai2021oort}. Other studies reduce computation and communication overhead through adaptive model pruning or jointly optimize pruning and resource-management decisions \cite{jiang2023prunefl,liu2022jointpruning}. More recently, model-heterogeneous federated learning has emerged as an alternative approach in which clients directly train models with different capacities. HeteroFL distributes sub-networks according to client capabilities \cite{diao2021heterofl}, while FjORD constructs nested sub-networks through ordered dropout \cite{horvath2021fjord}. FedRolex introduces rolling submodel extraction to improve parameter coverage across heterogeneous clients \cite{alam2022fedrolex}, and ScaleFL jointly adapts model width and depth for resource-constrained devices \cite{ilhan2023scalefl}. These studies demonstrate the potential of adapting model capacity to heterogeneous device capabilities.

Model heterogeneity has also begun to receive attention in serverless decentralized learning. The authors in \cite{tian2025communication} study communication-efficient personalized distributed learning under both data and node heterogeneity, with the goal of reducing communication overhead while accommodating heterogeneous participants. DisPFL employs personalized sparse masks and decentralized sparse training to reduce communication and computation costs \cite{dai2022dispfl}. DeSA facilitates knowledge transfer among clients with heterogeneous local models through synthetic anchors \cite{huang2024desa}. CompFL supports cooperative training of multiple predefined model architectures in both centralized and device-to-device settings \cite{xu2025compfl}, while DeMoFL further investigates peer-to-peer knowledge federation among model-heterogeneous mobile devices \cite{guo2026demofl}. 
The works in \cite{tian2026semantic, tian2025compositional} propose to transmit semantic information instead of model parameters to enable heterogeneous models among agents.
These methods mainly focus on enabling efficient training, personalization, aggregation, or knowledge transfer under heterogeneous decentralized settings. The question of how heterogeneous model capacities should be configured according to the deployment resources of individual nodes, however, remains comparatively less explored in fully decentralized learning.

Resource allocation provides another important approach to deal with system heterogeneity. In wireless federated learning, joint computation and communication resource management has been investigated by optimizing variables such as bandwidth, processing workload, and device computation speed \cite{zeng2021resource}. Other studies jointly consider model pruning, device selection, and wireless resource allocation to balance learning performance and training latency \cite{liu2022jointpruning}. Resource-aware designs have also been considered in decentralized settings. For example, grouped federated learning organizes devices with similar computation and communication capabilities to reduce training latency \cite{yin2022grouped}, while GossipFL adapts sparse peer-to-peer communication to improve bandwidth utilization \cite{tang2023gossipfl}. Nevertheless, these approaches primarily optimize communication, computation, scheduling, or topology-related decisions. They do not directly address the deployment-stage problem of jointly assigning node-specific model capacities in a decentralized system while simultaneously accounting for computation, communication, memory feasibility, and network-wide learning capability.

Motivated by this gap, we consider sub-network allocation, where the maximum model capacity deployed at each node is determined before decentralized training. The key challenge is to avoid resource bottlenecks on constrained nodes while still utilizing the capacity of resource-rich nodes and maintaining sufficient learning capability.
To address this problem, we propose Resource-aware Sub-network Allocation for Decentralized Learning (RASA-DL), a profiling-guided resource-aware sub-network allocation framework for heterogeneous decentralized learning. RASA-DL combines offline characterization of candidate sub-networks with current node- and link-level resource conditions to determine heterogeneous node-specific model capacities. The resulting allocation problem is formulated to minimize the synchronous bottleneck latency while satisfying memory and aggregate learning-utility requirements, and is solved efficiently by exploiting its monotone feasibility structure.
The main contributions of this work are summarized as follows.
\begin{itemize}

\item We propose RASA-DL, a profiling-guided resource-aware sub-network allocation framework for heterogeneous decentralized learning. RASA-DL treats sub-network size as a deployment variable and jointly considers heterogeneous computation, communication, memory, and learning requirements during configuring.

\item We develop a deployment-aware profiling methodology that characterizes candidate sub-networks from both system and learning perspectives. Offline hardware profiling captures the computation latency and memory consumption of candidate sub-networks, while reference decentralized environments profile their learning capability under the nested sub-network training and aggregation.

\item We formulate heterogeneous sub-network allocation as a synchronous bottleneck-latency minimization problem subject to memory and aggregate learning-utility constraints. By exploiting the monotone feasibility structure with respect to the latency threshold, we develop an efficient binary-search-based allocation algorithm that avoids exhaustive enumeration of heterogeneous sub-network assignments.

\item Extensive experiments are conducted under controlled system and statistical heterogeneity to evaluate the proposed framework from both learning and system perspectives. The results demonstrate that RASA-DL provides favorable accuracy--latency tradeoffs. Additional experiments validate the benefit of deployment-aware capability profiling and show that the heterogeneous deployments remain effective under different datasets and increasingly Non-IID data distributions.

\end{itemize}
The rest of this article is organized as follows. Section~\ref{sec2} introduces the decentralized learning system model and formulates the resource-aware sub-network allocation problem. Section~\ref{sec3} presents the offline profiling framework. Section~\ref{sec4} develops the proposed resource-aware allocation algorithm and describes the subsequent heterogeneous decentralized learning procedure. Section~\ref{sec5} presents the experimental setup and evaluates RASA-DL under different levels of system and statistical heterogeneity. Finally, Section~\ref{sec6} gives the conclusion.

\section{System Model and Problem Formulation} \label{sec2}
\subsection{Decentralized Learning}
Consider a decentralized learning network represented by an undirected graph $\mathcal{G}=(\mathcal{V},\mathcal{E})$, where $\mathcal{V}=\{1,\ldots,N\}$ denotes the set of participating nodes and $\mathcal{E}$ denotes the communication links. The neighbor set of node $i$ is denoted by $\mathcal{N}_i=\{j:(i,j)\in\mathcal{E}\}$. Each node $i$ owns a local dataset $\mathcal{D}_i$ and performs local model updates without exchanging raw data.

To accommodate heterogeneous device capabilities, we consider a canonical parent network $\mathcal M^{\mathrm p}$ and a finite family of structured compact sub-networks $\{\mathcal M(r):r\in\mathcal R\}$, where $\mathcal R=\{r^{(1)},\ldots,r^{(K)}\}$ denotes the candidate retention-ratio set with $0<r^{(1)}<\cdots<r^{(K)}\le 1$. Let $\Theta(r)$ denote the parameter set associated with $\mathcal M(r)$. The compact sub-networks are constructed in a nested manner such that
\begin{equation}
\Theta(r^{(1)})\subset\Theta(r^{(2)})\subset\cdots\subset\Theta(r^{(K)}).\label{eq:nested_parameter_sets}
\end{equation}

Accordingly, the parameters of the parent network can be partitioned into nested incremental blocks. By defining $\Theta(r^{(0)})=\varnothing$, the $k$th parameter block is
\begin{equation}
\mathcal B_k = \Theta(r^{(k)})\setminus\Theta(r^{(k-1)}), \qquad k=1,\ldots,K.
\label{eq:parameter_block}
\end{equation}
Hence, we have $\Theta(r^{(k)})=\bigcup_{q=1}^{k}\mathcal B_q$.

Each node $i$ is assigned a retention ratio $r_i\in\mathcal R$ according to its available system resources and executes the corresponding compact sub-network $\mathcal M(r_i)$. A node assigned a larger retention ratio therefore contains all parameter blocks of a node assigned a smaller retention ratio, together with additional higher-order blocks. This nested parameterization establishes a consistent parameter correspondence across heterogeneous sub-networks and enables block-wise decentralized aggregation: a parameter block is exchanged and aggregated only among nodes whose assigned sub-networks contain that block.
Note that retention ratio $r_i$ is treated as a system-level decision variable rather than a trainable model parameter. In general, a smaller $r_i$ reduces the computation, memory footprint, and communication payload of node $i$, whereas a larger $r_i$ preserves a higher-capacity sub-network. The resource-aware selection of $\{r_i\}_{i=1}^{N}$ is formulated in Section \ref{sec2_3}.

For a node whose retention ratio is $r_i$, local training does not necessarily activate the maximum-width sub-network at every local iteration. Instead, an active retention ratio $s\in\mathcal R_{\le r_i}$ is sampled according to
\begin{equation}
q_{\rho}(s\mid r)=\begin{cases}1,& r=r_{\min},~s=r_{\min},\\[2mm]\rho,& r>r_{\min},~s=r,\\[2mm]\dfrac{1-\rho}{|\mathcal R_{\le r}|-1},& r>r_{\min},~s<r,\\[3mm]0,& s>r,\end{cases}
\label{eq:nested_sampling_distribution}
\end{equation}
where $r_{\min}=\min\mathcal R$ and $\rho\in[0,1]$. Thus, $\rho$ represents the probability of activating the maximum deployed sub-network during a local training iteration, whereas the remaining probability mass is uniformly distributed among its nested lower-width sub-networks.

We consider synchronous decentralized learning. In each communication round, every node first performs local training using its assigned compact sub-network and subsequently exchanges model information with its neighboring nodes. A node proceeds to the next round only after the required neighboring updates have been received and processed. Consequently, heterogeneous computation and communication resources lead to different node completion times, and the synchronous training process is bottlenecked by the slowest participating node.

\subsection{Resource Models}
Within the network, each node $i$ can be characterized by a resource vector $\bm z_i=[c_i,M_i]$, where $c_i$ denotes its effective computation capability and $M_i$ denotes its available memory. The communication link $(i,j)$ is characterized by $\bm h_{ij}=[b_{ij},d_{ij}]$, where $b_{ij}$ denotes the effective bandwidth and $d_{ij}$ denotes the data-size-independent link latency, including propagation, queueing, and protocol-access delays.

\textbf{Computation latency. }
Let $\widehat T_i^{\mathrm{step}}(s)$ denote the profiled latency of one training step for active sub-network $\mathcal M(s)$ in node $i$, including forward propagation, backward propagation, and the optimizer update. The corresponding values are obtained from the profiling-based resource model described in Section~\ref{sec3}. 

Under the nested-width training policy, an active retention ratio $s\in\mathcal R_{\le r_i}$ is sampled according to $q_\rho(s\mid r_i)$. Therefore, with $K^{\mathrm{loc}}$ local training steps per communication round, the expected local training latency is
\begin{equation}
\widehat T_i^{\mathrm{tr}}(r_i;\rho)=K^{\mathrm{loc}}\sum_{s\in\mathcal R_{\le r_i}}q_\rho(s\mid r_i)\widehat T_i^{\mathrm{step}}(s).\label{eq:training_latency}
\end{equation}
Here, $r_i$ determines the maximum deployed sub-network, while $\rho$ controls the distribution of active sub-networks during local training.

Each node also performs local processing for overlap-aware aggregation of the received neighboring models. This cost depends on the neighboring sub-network configurations and is therefore not determined by $r_i$ alone. Our measurements show that the corresponding processing time contributes only approximately $1\%$ of the modeled round latency across the evaluated deployments. We therefore neglect this term and approximate the per-round computation latency as
\begin{equation}
\widehat T_i^{\mathrm{cmp}}(r_i;\rho)\approx\widehat T_i^{\mathrm{tr}}(r_i;\rho).\label{eq:computation_latency}
\end{equation}

\textbf{Communication latency. }
Let $P(r_i)$ denote the number of model parameters contained in the compact sub-network $\mathcal M(r_i)$. In contrast to unstructured pruning, $P(r_i)$ is generally not equal to $r_iP$, since structured channel pruning changes both the input and output dimensions of the retained layers. 

During decentralized training, node $i$ broadcasts the parameters of its maximum deployed compact sub-network $\mathcal M(r_i)$ to its neighboring nodes. Although smaller nested sub-networks may be activated during individual local training steps, they only affect the local computation workload and do not change the model exchanged at the end of the round. Since all candidate sub-networks belong to the same nested family, a receiving node can directly identify the parameter subset shared with its own local sub-network and perform aggregation only over the overlapping parameters.
Accordingly, the amount of information transmitted by node $i$ in one communication round can be modeled as $D_i(r_i)=b_w P(r_i)+D_i^{\mathrm{meta}}$, where $b_w$ denotes the number of bits used to represent one model parameter and $D_i^{\mathrm{meta}}$ denotes the additional metadata required to identify the corresponding compact-sub-network structure.
The communication latency from node $i$ to node $j$ is therefore
\begin{equation}
T_{ij}^{\mathrm{comm}}(r_i)=d_{ij}+\frac{D_i(r_i)}{b_{ij}}.
\end{equation}
Assuming that node $i$ communicates with all its neighbors in parallel and proceeds only after all outgoing transmissions are completed, its communication latency in one round is given by
\begin{equation}
\widehat T_i^{\mathrm{comm}}(r_i)=\max_{j\in\mathcal N_i}T_{ij}^{\mathrm{comm}}(r_i).\label{eq:communication_latency}
\end{equation}

\textbf{Memory consumption. }
The peak memory consumption of node $i$ is determined by its maximum deployed sub-network $\mathcal M(r_i)$. Let $S(r)$ denote the peak training memory of candidate sub-network $\mathcal M(r)$ obtained through offline hardware profiling. Since node $i$ must retain sufficient memory to execute its maximum assigned sub-network throughout decentralized training, its memory requirement is modeled as
\begin{equation}
\widehat S_i(r_i)=S(r_i).
\label{eq:profiled_memory}
\end{equation}
Accordingly, configuration $r_i$ is feasible at node $i$ only if
\begin{equation}
S(r_i)\leq M_i.
\end{equation}


\subsection{Resource-Aware sub-network Allocation} \label{sec2_3}
System heterogeneity can lead to severe synchronization inefficiency in decentralized learning. Assigning an identical model size to all nodes may cause resource-constrained devices to become persistent stragglers, while resource-rich nodes remain idle waiting for neighboring updates. We therefore determine the node-specific retention ratios before decentralized training to balance the synchronous round latency against the learning capability of the allocated sub-networks.

For each candidate retention ratio $r\in\mathcal R$, let $\widehat{\mathcal U}(r)$ denote its offline-profiled learning utility. The profiling procedure used to obtain $\widehat{\mathcal U}(r)$ is presented in Section~\ref{sec3_3}.

In parallel, the actual compact-network training workload is profiled on the target hardware. Together with the current computation and communication conditions, the completion time of node $i$ under configuration $r_i$ and $\rho$ is estimated as
\begin{equation}
\widehat{\tau}_i(r_i; \rho)=\widehat T_i^{\mathrm{cmp}}(r_i; \rho)+\widehat T_i^{\mathrm{comm}}(r_i). \label{eq:completion_time}
\end{equation}

\newtheorem{remark}{Remark}
\begin{remark}
Although the candidate sub-networks are structurally nested, i.e., $\mathcal M(r_a)\subseteq\mathcal M(r_b), r_a<r_b$, their empirical learning utilities are not necessarily monotonic with respect to $r$. Different compact architectures may exhibit different optimization and generalization behavior after training. We therefore directly use the profiled utility values over the finite candidate set without imposing a monotonicity assumption on $\widehat{\mathcal U}(r)$.
\end{remark}

Based on the profiled learning utility and the estimated resource costs, the resource-aware sub-network allocation problem can be formulated as
\begin{subequations}
\label{prob:sub-network_configuration}
\begin{align}
\min_{\bm r}\quad&\max_{i\in\mathcal V}\widehat{\tau}_i(r_i; \rho),\label{prob:configuration_obj}\\
\mathrm{s.t.}\quad&\widehat S_i(r_i)\le M_i,\quad\forall i\in\mathcal V,\label{prob:configuration_memory}\\
&\sum_{i=1}^{N}p_i\widehat{\mathcal U}(r_i)\ge U^{\mathrm{tar}},\label{prob:configuration_utility}\\
&r_i\in\mathcal R,\quad\forall i\in\mathcal V.\label{prob:configuration_set}
\end{align}
\end{subequations}
where $\bm r=[r_1,\ldots,r_N]$, $p_i\geq0$ denotes the contribution weight of node $i$ satisfying $\sum_{i=1}^{N}p_i=1$, and $U^{\mathrm{tar}}$ denotes the required aggregate learning utility. Here, $\rho$ is a prescribed system-level training parameter rather than an allocation variable. The optimization therefore determines only the node-specific maximum retention ratios $\bm r=[r_1,\ldots,r_N]$.

Problem \eqref{prob:sub-network_configuration} explicitly captures the tradeoff between system efficiency and learning capability. Smaller sub-networks can be assigned to resource-constrained nodes to alleviate computation, memory, and communication bottlenecks, while larger sub-networks can be retained by resource-rich nodes to satisfy the aggregate learning-utility requirement. The objective minimizes the modeled synchronous round bottleneck latency rather than the average node latency, thereby directly targeting the straggler effect in synchronous decentralized learning. In addition to improving the utilization of heterogeneous computing resources, a smaller and more stable round completion time improves the predictability and scalability of decentralized training and is expected to reduce the overall wall-clock training time in practical deployments.

\section{Offline Profiling Framework}\label{sec3}
The resource-aware allocation in Section \ref{sec2_3} relies on three types of prior information: a common family of nested compact sub-networks, the empirical learning utility associated with each candidate sub-network, and the resource profiles required to estimate its execution cost on heterogeneous devices. Rather than obtaining this information during decentralized training, we construct it through an offline profiling procedure, which is detailed below.

\subsection{Overall Offline Profiling Procedure}
The offline procedure consists of three stages. First, a canonical parent model $\mathcal M^{\mathrm p}$ is established as the common structural and parameter reference. A shared channel-importance ranking is then learned from the parent model and used to construct the nested compact-sub-network family $\{\mathcal M(r):r\in\mathcal R\}$. Second, the learning capability of each candidate width is evaluated in a lightweight reference decentralized environment to obtain the deployment-aware utility profile $\mathcal Q_{\mathrm U}$. Third, the executable sub-networks are benchmarked on reference hardware to obtain their training-latency and memory profiles.

These operations are performed before decentralized training and are independent of the subsequent node-specific allocation. In particular, the channel ranking, compact architectures, and learning-utility profile remain fixed after the offline stage. At deployment time, the offline information is combined with the current computation and communication conditions to evaluate the candidate configurations and solve the resource-aware allocation problem in \eqref{prob:sub-network_configuration}.

This offline-online separation is important for two reasons. First, architecture construction and capability evaluation do not need to be repeated whenever the system conditions change. Second, the allocation stage does not require trial execution of decentralized training under every possible combination of node-specific sub-networks, whose number grows combinatorially with the number of participating nodes.

\subsection{Nested Compact-sub-network Construction}
The construction starts from a canonical parent network $\mathcal M^{\mathrm p}$. Rather than independently pruning the parent network for each retention ratio, we learn auxiliary channel scores using labeled training data while keeping the parent parameters and BN statistics fixed. Channels are ranked within each residual stage
according to the aggregated $L_2$ norms of the learned score filters. All candidate subnetworks retain prefixes of the same stage-wise orderings, ensuring nested parameter alignment across model sizes and facilitating overlap-aware decentralized aggregation.

For each structured prunable layer $l\in\mathcal L_{\mathrm s}$ with $C_l$ output channels, let $\bm\pi_l=[\pi_l(1),\pi_l(2),\ldots,\pi_l(C_l)]$ denote the learned channel ordering from the most to the least important channel. Given a candidate retention ratio $r\in\mathcal R$, the retained output-channel set is obtained through prefix selection:
\begin{equation}
\left\{\pi_l(1),\ldots, \pi_l\left(\left\lceil rC_l\right\rceil\right)\right\}.
\end{equation}
The corresponding executable compact network associated with $r$ can be denoted by $\mathcal M(r)$ and induces the nested model family assumed in Section \ref{sec2}. 
The retained parameters are inherited from their corresponding locations in $\mathcal M^{\mathrm p}$. Unlike element-wise masking, this construction produces an executable network with reduced tensor dimensions. Therefore, the resulting architecture can provide actual reductions in computation, memory consumption, and communication payload rather than merely reducing the number of active weights.
The common ranking also provides a consistent parameter correspondence across different sub-networks. Specifically, if $r_a<r_b$, every parameter represented in $\mathcal M(r_a)$ has a corresponding parameter in $\mathcal M(r_b)$. This property is essential for heterogeneous decentralized learning because neighboring nodes with different retention ratios can identify and aggregate their overlapping parameters without introducing an additional architecture-matching procedure.
The common channel ordering is learned only once and remains fixed after the offline construction stage. 
Consequently, no channel-ranking or structural-pruning operation is performed during decentralized training. The resulting nested family $\mathcal M(r)$ constitutes the structural component of the offline sub-network library.

\subsection{Deployment-Aware Learning-Capability Profiling}
\label{sec3_3}
The retention ratio determines not only the resource cost of a compact sub-network, but also the learning performance that it can sustain during decentralized training. A straightforward approach is to evaluate each candidate sub-network independently and use its standalone validation accuracy as an architecture-level utility. However, as shown in Section \ref{sec5_3}, such an estimate may not accurately reflect deployment-time behavior and degrade the training performance. 
This discrepancy arises because, in the decentralized system, a compact sub-network is trained jointly with heterogeneous neighboring models and is affected by nested-width sampling and overlap-aware parameter aggregation.

To obtain a utility estimation that better reflects the intended deployment regime, we construct a lightweight \emph{reference decentralized environment}, denoted by $\mathcal E_{\mathrm{ref}}$. The reference environment contains a small set of decentralized nodes, denoted by $\mathcal V_{\mathrm{ref}}$, and includes representative maximum-width configurations from the candidate set $\mathcal R$. It uses the same nested sub-network construction, local training rule, width-sampling policy, normalization semantics, and overlap-aware decentralized aggregation as the target system, while requiring substantially less computation than repeatedly evaluating full-scale deployments.

Let
\begin{equation}
\mathcal V_r^{\mathrm{ref}}=\left\{i\in\mathcal V_{\mathrm{ref}}: r_i=r\right\}
\end{equation}
denote the set of reference nodes whose maximum deployed width is $r$. For profiling run $j\in{1,\ldots,J}$, after a fixed budget of decentralized training, the learning capability associated with candidate width $r$ is measured as
\begin{equation}
\widehat{\mathcal U}_j(r)=\frac{1}{|\mathcal V_r^{\mathrm{ref}}|}\sum_{i\in\mathcal V_r^{\mathrm{ref}}}A^{\mathrm{val}}_{i,j},\label{eq:reference_environment_utility}
\end{equation}
where $A^{\mathrm{val}}_{i,j}$ denotes the validation accuracy of reference node $i$ at the end of profiling run $j$.

The deployment-aware learning utility of candidate width $r$ is then defined as the empirical mean over the $J$ profiling runs:
\begin{equation}
\widehat{\mathcal U}(r)=\frac{1}{J}\sum_{j=1}^{J}\widehat{\mathcal U}_j(r).\label{eq:profiled_utility}
\end{equation}
Accordingly, the resulting learning-capability profile is
\begin{equation}
\mathcal Q_{\mathrm U}=\left\{\left(r,\widehat{\mathcal U}(r)\right):r\in\mathcal R\right\}.
\label{eq:utility_profile}
\end{equation}

The purpose of $\mathcal E_{\mathrm{ref}}$ is not to reproduce a particular target deployment exactly. Instead, it provides a low-cost calibration environment in which every candidate width is exposed to the decentralized interactions that are absent from standalone evaluation. In particular, the profiling procedure captures the effects of training alongside heterogeneous maximum-width sub-networks, nested active-width sampling, and parameter exchange over overlapping model coordinates.

The resulting quantity $\widehat{\mathcal U}(r)$ should therefore be interpreted as a deployment-aware, configuration-level capability estimation, rather than as a prediction of the final accuracy of an individual node. The latter depends on the complete distributed optimization trajectory, local data, neighboring configurations, and the realized network environment. By contrast, $\widehat{\mathcal U}(r)$ provides a common scalar utility associated with candidate width $r$ and is used only to guide the subsequent resource-allocation problem.
Once constructed, $\mathcal Q_{\mathrm U}$ remains fixed during resource allocation and decentralized training. Importantly, the profiling procedure preserves the separable structure of the allocation problem: each candidate width is still represented by a single scalar utility value, and therefore no modification to the feasibility characterization or latency-threshold search is required.

Finally, the proposed profiling procedure does not impose any analytical or monotonic relationship between retention ratio and learning capability. Although larger sub-networks generally provide greater representational capacity, their realized performance under decentralized training may depend on the interaction dynamics. The allocation algorithm therefore operates directly on the empirically profiled values in $\mathcal Q_{\mathrm U}$, allowing arbitrary learning-capability profiles over the finite candidate set.

\begin{figure*}[t]
    \centering
    \includegraphics[width=\textwidth]{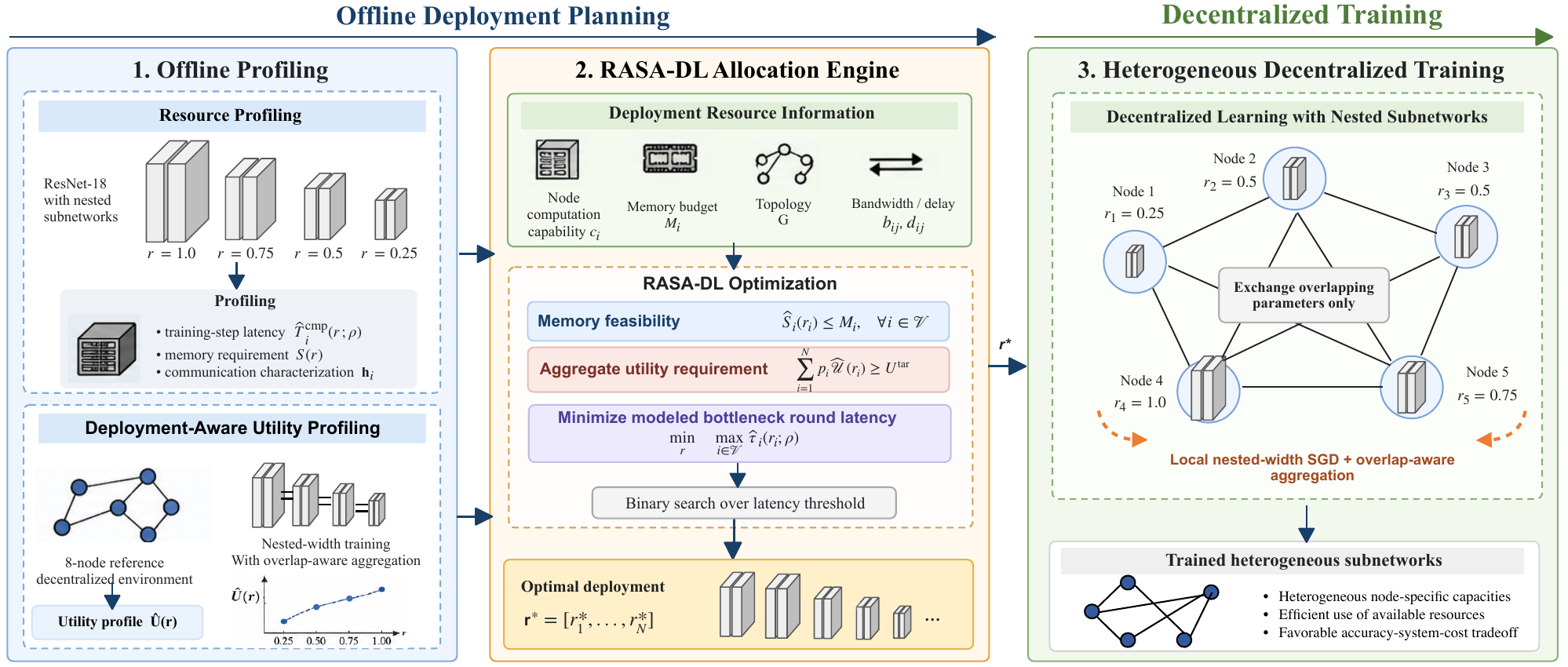}
    \caption{Overview of RASA-DL. Offline profiling characterizes the resource cost and deployment-aware learning utility of nested candidate sub-networks. Given node- and link-level resource information, RASA-DL determines heterogeneous node-specific sub-network capacities by minimizing the modeled synchronous bottleneck cost under memory and utility constraints. The resulting sub-networks are then trained through nested decentralized learning with overlap-aware parameter exchange and aggregation.}
    \label{fig:framework}
\end{figure*}

\subsection{Hardware Resource Profiling}
In addition to learning capability, resource-aware allocation requires the execution cost associated with each candidate compact sub-network. Analytical quantities such as parameter count or floating-point operations alone may not accurately reflect wall-clock training cost because actual execution depends on tensor dimensions, memory-access behavior, and hardware efficiency. We therefore characterize candidate sub-networks through empirical hardware profiling.

For each retention ratio $r\in\mathcal R$, the compact sub-network $\mathcal M(r)$ is executed on the reference hardware using the same basic training workload adopted during decentralized learning. After an initial warm-up period, we measure the latency of a complete training step, including forward propagation, backward propagation, and the optimizer update. The resulting computation profile is denoted by
\begin{equation}
\mathcal Q_{\mathrm T}=\left\{\left(r,\overline T^{\mathrm{step}}(r)\right):r\in\mathcal R\right\}.
\end{equation}

Peak training memory is measured in the same profiling process, yielding
\begin{equation}
\mathcal Q_{\mathrm S}=\left\{\left(r,S(r)\right):r\in\mathcal R\right\}.
\end{equation}
Since the candidate set is finite, these profiled values can be stored directly and reused during subsequent deployment.

Communication cost does not require empirical prediction. For every candidate sub-network, its parameter count $P(r)$ and hence its communication payload are known once the compact architecture is materialized. Given the current link bandwidth and propagation delay, the communication latency can therefore be evaluated directly using the model in Section~\ref{sec2}.

After offline profiling, each candidate retention ratio is associated with
\begin{equation}
\left\{\mathcal M(r),P(r),\widehat{\mathcal U}(r),\overline T^{\mathrm{step}}(r),S(r)\right\}.
\end{equation}
These quantities form a reusable profiling library that is combined with the current node and link conditions during resource-aware sub-network allocation.

\section{Algorithm Development}\label{sec4}
The overall RASA-DL framework is summarized in Fig.~\ref{fig:framework}.
We next derive an efficient solution to Problem~\eqref{prob:sub-network_configuration}.

\subsection{Algorithm Derivation}
Before solving Problem~\eqref{prob:sub-network_configuration}, we first remove node-specific dominated configurations.  For node $i$, each candidate configuration can be characterized by $\left(\widehat{\tau}_i(r_i; \rho),\widehat S_i(r_i),\widehat{\mathcal U}(r_i)\right)$. We define that configuration $r_b$ is dominated by $r_a$ for node $i$ if
\begin{subequations}
\label{eq:dominance_conditions}
\begin{align}
\widehat{\tau}_i(r_a; \rho)&\le\widehat{\tau}_i(r_b; \rho),\\
\widehat S_i(r_a)&\le\widehat S_i(r_b),\\
\widehat{\mathcal U}(r_a)&\ge\widehat{\mathcal U}(r_b),
\end{align}
\end{subequations}
with at least one strict inequality. In this case, replacing $r_b$ by $r_a$ cannot increase the completion time or memory consumption and does not reduce the profiled learning utility. Hence, $r_b$ can be removed without affecting optimality. The remaining configurations constitute the node-specific Pareto set $\mathcal R_i^{\mathrm P}$.

After Pareto pruning, the min--max problem in \eqref{prob:sub-network_configuration} can be solved through a sequence of latency-feasibility tests. For a candidate bottleneck latency threshold $T$, the objective is to determine whether there exists a sub-network-allocation vector satisfying all resource and learning constraints.

For a given $T$, define the feasible configuration set of node $i$ as
\begin{equation}
\mathcal F_i(T)=\left\{r\in\mathcal R_i^{\mathrm P}\;\middle|\;\widehat{\tau}_i(r;\rho)\le T,\;\widehat S_i(r)\le M_i\right\}.
\label{eq:configuration_feasible_set}
\end{equation}
Note that the learning utility of each candidate compact sub-network has already been obtained through offline profiling as $\widehat{\mathcal U}(r)$ in \eqref{eq:profiled_utility}. For any $T$ such that $\mathcal F_i(T)\neq\varnothing$, define the maximum learning utility achievable by node $i$ under latency threshold $T$ as
\begin{equation}
\overline{\mathcal U}_i(T)=\max_{r\in\mathcal F_i(T)}\widehat{\mathcal U}(r).
\label{eq:max_utility}
\end{equation}
Let
\begin{equation}
\widetilde r_i(T)\in\arg\max_{r\in\mathcal F_i(T)}\widehat{\mathcal U}(r)
\label{eq:best_configuration_under_T}
\end{equation}
denote one utility-maximizing sub-network configuration.

\newtheorem{proposition}{Proposition}
\newtheorem{theorem}{Theorem}
\begin{proposition}\label{p1}
(Feasibility Characterization) A latency threshold $T$ is feasible for Problem \eqref{prob:sub-network_configuration} if and only if
\begin{equation}
\mathcal F_i(T)\neq\varnothing,\quad\forall i\in\mathcal V,\label{eq:nonempty_feasible_sets}
\end{equation}
and
\begin{equation}
\overline{\mathcal U}_{\mathrm{sys}}(T)=\sum_{i=1}^{N}p_i\overline{\mathcal U}_i(T)\ge U^{\mathrm{tar}}.\label{eq:utility_feasibility}
\end{equation}
\end{proposition}

\begin{proof}
If $\mathcal F_i(T)=\varnothing$ for any $i$, then $T$ is infeasible. 
Otherwise, since the feasible allocation set is the Cartesian product $\prod_{i=1}^{N}\mathcal F_i(T)$ and the aggregate utility objective is additively separable across nodes, the maximum aggregate utility achievable under $T$ is
\begin{equation}
\overline{\mathcal U}_{\mathrm{sys}}(T)=\max_{\substack{r_i\in\mathcal F_i(T),\\i\in\mathcal V}}\sum_{i=1}^{N}p_i\widehat{\mathcal U}(r_i)=\sum_{i=1}^{N}p_i\overline{\mathcal U}_i(T).\label{eq:maximum_aggregate_utility}
\end{equation}
Therefore, $T$ is feasible if and only if \eqref{eq:nonempty_feasible_sets} and \eqref{eq:utility_feasibility} hold.
\end{proof}

\begin{theorem}\label{t1}
(Monotone Feasibility) The feasibility characterized in Proposition \ref{p1} is monotonic with respect to the latency threshold $T$. 
Consequently, the feasible thresholds form an upper interval over the ordered candidate set $\mathcal T$, and the minimum feasible threshold can be obtained by binary search.
\end{theorem}

\begin{proof}
Assume $T_1$ is feasible and $T_2\ge T_1$. Note that $\mathcal F_i(T_1)\subseteq\mathcal F_i(T_2)$, and consequently
\begin{equation}
\overline{\mathcal U}_i(T_1)\le\overline{\mathcal U}_i(T_2).
\end{equation}
Hence, $\overline{\mathcal U}_{\mathrm{sys}}(T)$ is non-decreasing in $T$. Therefore, the feasible thresholds form an upper interval, which completes the proof.
\end{proof}

To exploit the monotone feasibility established in Theorem \ref{t1}, the candidate latency set is constructed as
\begin{equation}
\mathcal T=\left\{\widehat{\tau}_i(r, \rho)\,\middle|\,i\in\mathcal V,\;r\in\mathcal R_i^{\mathrm P},\;\widehat S_i(r)\le M_i\right\}.\label{eq:candidate_latency_set}
\end{equation}
By Proposition~\ref{p1} and Theorem~\ref{t1}, the optimal bottleneck latency is given by
\begin{equation*}
T^\star=\min\{T\in\mathcal T:T\text{ is feasible}\}.
\end{equation*}
The corresponding optimal sub-network allocation can be obtained by
\begin{equation*}
r_i^\star=\widetilde r_i(T^\star),\quad\forall i\in\mathcal V.
\end{equation*}
which satisfies the utility constraint and achieves the minimum bottleneck latency.

\subsection{Resource-aware Sub-network Allocation}
The algorithm can be summarized as in Algorithm \ref{alg:retention_allocation}.
\begin{algorithm}[!htp]
\caption{\textbf{Resource-Aware sub-network Allocation for Decentralized Learning (RASA-DL)}}\label{alg:retention_allocation}
\textbf{Input:} Pareto sets $\{\mathcal R_i^{\mathrm P}\}_{i=1}^{N}$, weights $\{p_i\}$, memory budgets $\{M_i\}$, target utility $U^{\mathrm{tar}}$, and fixed sampling parameter $\rho$.\\
Construct $\mathcal T$ according to \eqref{eq:candidate_latency_set} and sort it in ascending order.\\
\If{$\mathcal T=\varnothing$}{
    \textbf{Return} infeasible
}
Set $L\leftarrow1$ and $H\leftarrow|\mathcal T|$.\\
\If{$\mathcal T[H]$ is infeasible according to \eqref{eq:nonempty_feasible_sets}--\eqref{eq:utility_feasibility}}{
    \textbf{Return} infeasible
}
\While{$L<H$}{
    $q\leftarrow\lfloor(L+H)/2\rfloor$ and $T\leftarrow\mathcal T[q]$.\\
    \If{$T$ satisfies \eqref{eq:nonempty_feasible_sets} and \eqref{eq:utility_feasibility}}{
        $H\leftarrow q$
    }
    \Else{
        $L\leftarrow q+1$
    }
}
$T^\star\leftarrow\mathcal T[L]$ and $r_i^\star\leftarrow\widetilde r_i(T^\star)$ for all $i\in\mathcal V$.\\
\textbf{Return} the optimal bottleneck latency $T^\star$ and the sub-network allocation $\mathbf r^\star$.
\end{algorithm}
The proposed optimization operates only on configuration-level latency, memory, and learning-utility values obtained from the offline profiles and current system conditions, without executing decentralized training during allocation.
Let
\begin{equation}
K_{\mathrm P}=\sum_{i=1}^{N}|\mathcal R_i^{\mathrm P}|
\end{equation}
denote the total number of retained configurations. Constructing and sorting the candidate latency set requires $\mathcal O(K_{\mathrm P}\log K_{\mathrm P})$ time. Each feasibility test scans at most $K_{\mathrm P}$ configurations, and the binary search performs $\mathcal O(\log K_{\mathrm P})$ tests. Given the node-specific Pareto sets, the allocation procedure has a complexity of $\mathcal O(K_{\mathrm P}\log K_{\mathrm P})$. Since all profiling procedures are completed offline only once, the online optimization introduces negligible computational overhead compared with decentralized model training.

\subsection{Nested Compact Decentralized Learning}
After the resource-aware allocation is completed, each node $i$ is assigned a maximum retention ratio $r_i^\star$ and initializes the corresponding compact sub-network $\mathcal M(r_i^\star)$ from the canonical parent model. The assigned ratio remains fixed throughout decentralized training.

Following the nested-width sampling strategy of FjORD \cite{horvath2021fjord}, node $i$ samples an active retention ratio at each local iteration as
\begin{equation}
s_i^{t,k}\sim q_{\rho}(\cdot\mid r_i^\star),\qquad s_i^{t,k}\in\mathcal R_{\le r_i^\star},
\end{equation}
and performs SGD updates only on the parameters of the sampled sub-network $\Theta(s_i^{t,k})$. Parameters outside the active sub-network remain unchanged during that local iteration.

After $K^{\mathrm{loc}}$ local updates, each node exchanges the parameters of its maximum deployed sub-network with its neighbors. Since the candidate sub-networks are nested, aggregation is performed only over overlapping parameters. For parameter block $\mathcal B_q\subseteq\Theta(r_i^\star)$, define
\begin{equation}
\mathcal A_i^{(q)}=\{i\}\cup\left\{j\in\mathcal N_i:\mathcal B_q\subseteq\Theta(r_j^\star)\right\}.
\end{equation}
Node $i$ then performs
\begin{equation}
\bm\theta_{i,\mathcal B_q}^{t+1}=\sum_{j\in\mathcal A_i^{(q)}}w_{ij}^{(q)}\bm\theta_{j,\mathcal B_q}^{t,K^{\mathrm{loc}}},
\end{equation}
where the mixing weights satisfy
\begin{equation}
\sum_{j\in\mathcal A_i^{(q)}}w_{ij}^{(q)}=1.
\end{equation}

The resulting training procedure is summarized in Algorithm~\ref{alg:nested_decentralized_training}. The canonical parent model is used only for initialization, while subsequent decentralized training operates exclusively on the allocated compact sub-networks.
\begin{algorithm}[!htp]
\caption{\textbf{Nested Compact Decentralized Learning}}
\label{alg:nested_decentralized_training}
\textbf{Input:} Graph $\mathcal G$, allocated widths $\{r_i^\star\}$, sampling distribution $q_\rho$, local steps $K^{\mathrm{loc}}$, communication rounds $K^{\mathrm{rnd}}$, and learning rate $\eta$.\\
Initialize each node $i$ with $\mathcal M(r_i^\star)$.\\
\For{$t=0,\ldots,K^{\mathrm{rnd}}-1$}{
    \ForEach{node $i\in\mathcal V$ \textbf{in parallel}}{
        \For{$k=0,\ldots,K^{\mathrm{loc}}-1$}{
            Sample $s_i^{t,k}\sim q_\rho(\cdot\mid r_i^\star)$.\\
            Perform one local SGD update on $\mathcal M(s_i^{t,k})$.
        }
        Broadcast the parameters of $\mathcal M(r_i^\star)$ to neighboring nodes.
    }
    \ForEach{node $i\in\mathcal V$ \textbf{in parallel}}{
        Aggregate each shared parameter block over $\mathcal A_i^{(q)}$.
    }
}
\textbf{Return} trained compact models $\{\bm\theta_i\}_{i=1}^{N}$.
\end{algorithm}

\section{Experiment}\label{sec5}
\subsection{Experimental Settings}
We first construct a nested family of sub-networks from a common canonical parent model. The parent model is initialized from an ImageNet-pretrained ResNet-18 and adapted to CIFAR-10 through full-model fine-tuning. A fixed structured channel ranking is then obtained from the parent model and shared by all candidate sub-networks. Based on this ranking, four nested configurations with retention ratios $\mathcal{R}=\{1.0,0.75,0.5,0.25\}$ are constructed, corresponding to stage widths of $[64,128,256,512]$, $[48,96,192,384]$, $[32,64,128,256]$, and $[16,32,64,128]$, respectively. Their parameter counts are approximately $11.18$, $6.29$, $2.80$, and $0.70$ million. All sub-networks inherit the corresponding initialization parameters from the same canonical parent and preserve a nested parameter structure.

To estimate the learning capability of each candidate width, we use the deployment-aware profiling procedure described in Section~\ref{sec3_3}.
The reference environment consists of eight decentralized nodes connected by a fixed ring topology. The maximum widths $\{0.25,0.5,0.75,1.0\}$ are assigned twice each across the eight virtual nodes. Profiling uses an IID partition of CIFAR-10, with 45,000 samples used for decentralized training and 5,000 held out for validation. Each profiling run lasts 30 communication rounds with $100$ local optimization steps per node and a batch size of $64$. 
Nested-width sampling uses $\rho=0.7$, together with the same overlap-aware aggregation rule employed in the main decentralized training. The round-30 validation accuracies are averaged across nodes with the same maximum width and across three training seeds. This yields the frozen deployment-aware utilities, which are used only as profiling-based capability estimates for deployment optimization and should not be interpreted as the final test accuracies of the system-level experiments. Unless otherwise stated, the CIFAR-10 training data are partitioned IID across the 20 nodes with equal sample sizes.

The execution characteristics of each candidate sub-network are profiled on an NVIDIA GH200 GPU, which serves as the reference hardware platform. A complete training step includes forward propagation, backward propagation, and the optimizer update. With a batch size of $64$, the measured step latencies for retention ratios $1.0$, $0.75$, $0.5$, and $0.25$ are $15.42$, $14.31$, $12.83$, and $11.42$ ms, respectively, averaged over three profiling runs after warm-up. The corresponding training-memory requirements are approximately $1.629$, $1.222$, $0.845$, and $0.471$ GB. These measurements provide the empirical execution profiles used as the basis of the subsequent system-level evaluation.

Since a physical testbed containing a sufficiently diverse set of heterogeneous devices and network links is not available, we evaluate RASA-DL in a controlled simulated decentralized environment. Specifically, we consider $20$ nodes connected by a fixed undirected graph with $40$ communication edges and an average node degree of $4$. Node-level computation costs are derived from the measured reference-hardware profiles and adjusted according to simulated computation capabilities, while memory budgets and link-level bandwidths and propagation delays are independently configured. Communication costs are then evaluated using the analytical model in Section~\ref{sec2}. Consequently, the reported system metric is a \emph{modeled bottleneck round latency}, rather than measured end-to-end wall-clock latency on a physical heterogeneous testbed.

To systematically study the effect of system heterogeneity, we construct three resource regimes, denoted as \emph{Mild}, \emph{Moderate}, and \emph{Severe}. The same topology and mean resource levels are maintained across the three regimes, while only the dispersion of node and link resources is varied. Specifically, the arithmetic mean computation slowdown factor, memory budget, link bandwidth, and propagation delay are fixed at $2.5$, $1.3$ GB, $0.65$ Gbps, and $10.5$ ms, respectively. The computation slowdown factors have coefficients of variation (CVs) of $0.119$, $0.349$, and $0.450$ for the Mild, Moderate, and Severe regimes; the corresponding memory-budget CVs are $0.133$, $0.358$, and $0.400$, and the bandwidth CVs are $0.109$, $0.449$, and $0.550$. We summarize the overall degree of resource heterogeneity as
\begin{equation}
H=\sqrt{\frac{\mathrm{CV}_{\mathrm{comp}}^2+\mathrm{CV}_{\mathrm{mem}}^2+\mathrm{CV}_{\mathrm{bw}}^2}{3}},
\end{equation}
which gives $H=0.121$, $0.388$, and $0.471$, respectively. In the simulation, the computation capability of node $i$ is instantiated through a node-specific slowdown factor $\gamma_i$, which scales the profiled reference latency of each active sub-network, while the memory constraint is evaluated as $S(r_i)\leq M_i$. Under the resulting resource realizations, the largest globally feasible uniform width is $0.5$ in the Mild regime and $0.25$ in both the Moderate and Severe regimes.

\subsection{Decentralized Learning Performance}
\label{sec5_2}
We next evaluate the decentralized learning performance of RASA-DL under the three controlled resource regimes. All methods use the same decentralized topology, data partition, communication-round budget, and evaluation protocol. Unless otherwise stated, training is conducted for $50$ communication rounds, with $100$ local optimization steps per node in each round, a batch size of $64$, and SGD with a learning rate of $0.001$, momentum of $0.9$, and weight decay of $5\times10^{-4}$. Each configuration is evaluated using three random seeds, and the weighted mean test accuracy across all nodes is reported. System efficiency is characterized by the modeled bottleneck round latency determined by the compute and communication costs of the corresponding deployment. We define the modeled bottleneck round latency as the maximum,
over all nodes, of the sum of local-training compute latency and parallel-neighbor communication latency.

\paragraph{Baselines and RASA-DL operating points}
We consider both system-level baselines and an existing decentralized learning method. \emph{Fixed Uniform} assigns the same fixed-width sub-network to all nodes and trains only that width, whereas \emph{Nested Uniform} uses a common maximum width but adopts the same nested-width sampling mechanism as RASA-DL during local training. Both methods aggregate through averaging.
We further include DFedAvgM~\cite{sun2023decentralized} as a literature baseline for decentralized federated learning. Its method-specific momentum persistence and neighbor-mixing rules are preserved, while the deployed model width is chosen as the largest width that is feasible on every node. Accordingly, the maximum feasible uniform width is $0.5$ in the Mild regime and $0.25$ in the Moderate and Severe regimes.

For RASA-DL, we select three representative operating points from the deployment frontier under different utility targets. RASA-DL-A denotes the minimum-latency endpoint, RASA-DL-B corresponds to the common intermediate profiling target $U^{\mathrm{tar}}=0.89$, and RASA-DL-C denotes the maximum-profiled-utility endpoint. The allocation of RASA-DL-C coincides with the Memory-Max assignment, which independently selects the largest memory-feasible sub-network at each node. Hence, RASA-DL-C also serves as the resource-only maximum-capacity reference without requiring a duplicate training run. For RASA-DL and Nested Uniform, the nested-width sampling parameter is fixed to $\rho=0.7$.

\paragraph{Accuracy--latency tradeoff}
Fig.~\ref{fig:accuracy_latency} presents the primary system-level comparison by jointly considering final learning performance and bottleneck latency. For each resource regime, we report the maximum-feasible Fixed Uniform and Nested Uniform deployments, DFedAvgM, and the three representative RASA-DL operating points. Each point denotes the three-seed mean weighted node accuracy, with vertical error bars indicating one sample standard deviation. The baseline methods are shown as isolated operating points, whereas RASA-DL-A, RASA-DL-B, and RASA-DL-C are connected to visualize the accuracy--latency frontier enabled by deployment-aware heterogeneous allocation.

\begin{figure*}[t]
    \centering
    \includegraphics[width=\textwidth]{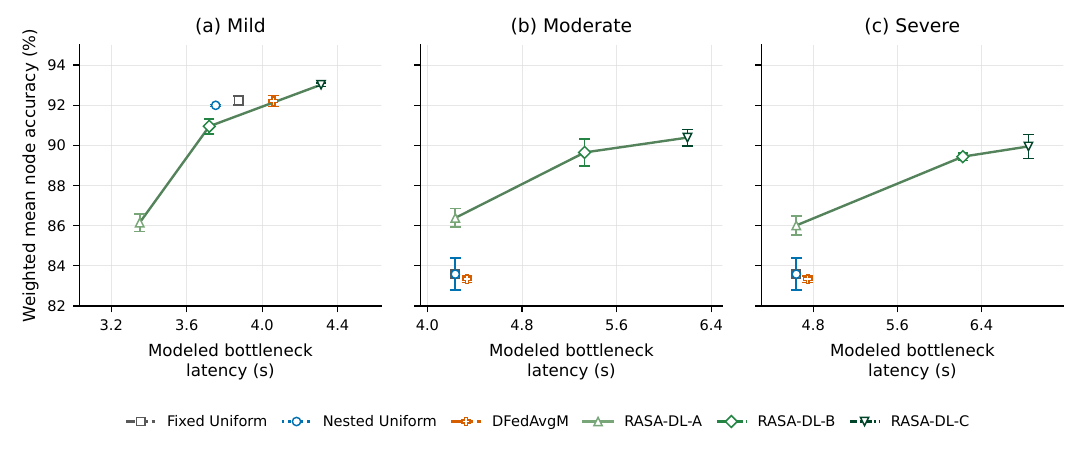}
    \caption{Accuracy--latency tradeoff under (a) Mild, (b) Moderate, and (c) Severe resource heterogeneity. Each point reports the three-seed mean weighted node accuracy, with vertical error bars indicating $\pm1$ sample standard deviation. Baselines are shown as isolated operating points, while RASA-DL-A, RASA-DL-B, and RASA-DL-C are connected to represent the RASA-DL deployment frontier. The maximum feasible uniform width is $0.5$ in Mild and $0.25$ in Moderate and Severe. The panels use scenario-specific latency ranges for readability while sharing the same accuracy scale.}
    \label{fig:accuracy_latency}
\end{figure*}

The results exhibit a clear change in behavior across the three resource regimes. Under Mild heterogeneity, every node can accommodate the $0.5$-width sub-network. Consequently, the maximum-feasible uniform deployments and DFedAvgM can employ a relatively expressive common model and remain competitive with the RASA-DL frontier. In this regime, RASA-DL mainly provides a broader set of deployment choices spanning lower latency and higher learning capability, rather than uniformly dominating the strongest common-model deployment.

The situation changes under Moderate and Severe heterogeneity. Although the mean primitive resource levels remain unchanged, increased resource dispersion reduces the largest globally feasible uniform width to $0.25$. Uniform-model methods must therefore restrict every node to the smallest candidate sub-network, even though many individual nodes can accommodate larger models. 
In both regimes, RASA-DL-A achieves higher accuracy than both the maximum-feasible uniform deployment and DFedAvgM at the same modeled bottleneck latency. Moving from RASA-DL-A to RASA-DL-B and RASA-DL-C further improves learning performance at the cost of additional latency, yielding multiple non-dominated operating points rather than a single fixed deployment.

The corresponding numerical results are summarized in Table~\ref{tab:main_results}. At the minimum-latency operating point, RASA-DL-A improves the latency-matched maximum-feasible uniform deployment by $2.8017$ and $2.4265$ percentage points in the Moderate and Severe regimes, respectively. Compared with DFedAvgM, RASA-DL-A achieves gains of $3.0505$ and $2.6753$ percentage points, respectively, at the same modeled bottleneck round latency. In Mild, by contrast, the maximum-feasible uniform width remains $0.5$, and both Fixed Uniform and DFedAvgM achieve accuracies above $92\%$, confirming that a sufficiently expressive common deployment remains competitive.

\begin{table*}[t]
\centering
\caption{Decentralized learning performance under controlled resource heterogeneity. Accuracy is the weighted mean node test accuracy over
three random seeds and is reported as mean $\pm$ sample standard deviation. Deployment is reported as the number of nodes assigned to retention ratios $[0.25,0.50,0.75,1.00]$. Latency denotes the modeled bottleneck round latency, including local training computation and inter-node model communication.}
\label{tab:main_results}
\setlength{\tabcolsep}{7pt}
\renewcommand{\arraystretch}{1.08}
\begin{tabular}{llccc}
\toprule
Regime & Method & Deployment Counts & Accuracy (\%) & Latency (s) \\
\midrule
Mild & Fixed Uniform & $[0,20,0,0]$ & $92.2272 \pm 0.2262$ & 3.8754 \\
& Nested Uniform & $[0,20,0,0]$ & $91.9955 \pm 0.0508$ & 3.7536 \\
& DFedAvgM & [0,20,0,0] & $92.2042 \pm 0.2699$ & 3.8754 \\
& RASA-DL-A & $[9,7,4,0]$ & $86.1517 \pm 0.4289$ & 3.3506 \\
& RASA-DL-B & $[2,13,5,0]$ & $90.9498 \pm 0.3730$ & 3.7196 \\
& RASA-DL-C & $[0,8,12,0]$ & $93.0220 \pm 0.0905$ & 4.3120 \\

\midrule
Moderate & Fixed Uniform & $[20,0,0,0]$ & $83.5870 \pm 0.8011$ & 4.2358 \\
& Nested Uniform & $[20,0,0,0]$ & $83.5870 \pm 0.8011$ & 4.2358 \\
& DFedAvgM & [20,0,0,0] & $83.3382 \pm 0.1582$ & 4.2358 \\
& RASA-DL-A & $[11,4,1,4]$ & $86.3887 \pm 0.4576$ & 4.2358 \\
& RASA-DL-B & $[5,7,4,4]$ & $89.6477 \pm 0.6701$ & 5.3293 \\
& RASA-DL-C & $[5,4,4,7]$ & $90.3877 \pm 0.4117$ & 6.1990 \\

\midrule
Severe & Fixed Uniform & $[20,0,0,0]$ & $83.5870 \pm 0.8011$ & 4.6387 \\
& Nested Uniform & $[20,0,0,0]$ & $83.5870 \pm 0.8011$ & 4.6387 \\
& DFedAvgM & [20,0,0,0] & $83.3382 \pm 0.1582$ & 4.6387 \\
& RASA-DL-A & $[12,3,1,4]$ & $86.0135 \pm 0.4758$ & 4.6387 \\
& RASA-DL-B & $[6,4,6,4]$ & $89.4388 \pm 0.1789$ & 6.2222 \\
& RASA-DL-C & $[6,3,4,7]$ & $89.9483 \pm 0.5888$ & 6.8466 \\
\bottomrule
\end{tabular}
\vspace{-0.5cm}
\end{table*}

\paragraph{Convergence behavior}
To further examine whether the final performance differences arise from stable decentralized training, Fig.~\ref{fig:convergence} compares the learning trajectories under the three resource regimes. To keep the comparison focused, we show the maximum-feasible Fixed Uniform deployment, DFedAvgM, and the two endpoint configurations of the RASA-DL frontier, RASA-DL-A and RASA-DL-C. The curves report the three-seed mean weighted node accuracy at each communication round.

Across all three regimes, the evaluated configurations exhibit stable learning behavior over the $50$ communication rounds. RASA-DL-A and RASA-DL-C converge consistently to their respective final operating points, with RASA-DL-C attaining higher final accuracy due to its higher-capacity deployment. The trajectories also show that the final differences in Fig.~\ref{fig:accuracy_latency} emerge progressively during decentralized training rather than from transient fluctuations at the final round.

\begin{figure*}[t]
    \centering
    \includegraphics[width=\textwidth]{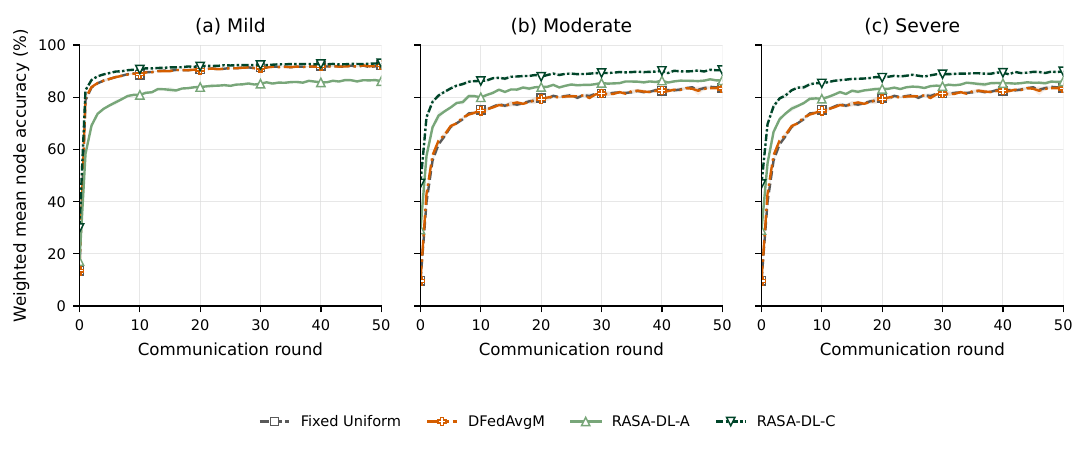}
    \caption{Convergence behavior under (a) Mild, (b) Moderate, and (c) Severe resource heterogeneity. Curves show the mean weighted node accuracy over three random seeds, and shaded regions indicate $\pm1$ sample standard deviation. Fixed Uniform and DFedAvgM use the largest globally feasible uniform width in each regime, i.e., $0.5$ for Mild and $0.25$ for Moderate and Severe. RASA-DL-A and RASA-DL-C represent the minimum-latency and maximum-profiled-utility endpoints of the RASA-DL deployment frontier, respectively.}
    \label{fig:convergence}
    \vspace{-0.5cm}
\end{figure*}

\paragraph{Discussion}
Taken together, the results indicate that the benefit of heterogeneous deployment is governed less by the scalar heterogeneity index $H$ itself than by the capacity gap induced by heterogeneous resource constraints. In particular, RASA-DL becomes most useful when the largest model that is feasible for all nodes is substantially smaller than the models supported by many individual nodes. By exploiting this otherwise unused capacity, RASA-DL converts device-level resource diversity into a configurable accuracy--latency tradeoff instead of forcing the entire decentralized system to operate at the capability of its least-capable participant.

\subsection{Comparison Between Standalone and Deployment-Aware Profiling}
\label{sec5_3}
We further compare the learning trajectories of the deployments selected by standalone and deployment-aware capability profiling under the same
utility requirement. Specifically, we fix the utility target to $U^{\mathrm{tar}}=0.89$, corresponding to the RASA-DL-B experiments. Under this common target, standalone and deployment-aware profiling select different deployments in all three controlled resource regimes. The subsequent decentralized training procedure is identical in the two cases; hence, the only difference originates from the capability profile used for deployment selection.

Fig.~\ref{fig:profiling_convergence} reports the weighted mean node accuracy as a function of cumulative modeled latency. Each curve shows the mean over three random seeds, with the shaded region denoting $\pm1$ sample standard deviation. The horizontal dashed line indicates the common requested learning-capability target of $89\%$. 
The two profiling strategies lead to clearly different learning outcomes under the same requested capability. Standalone profiling assigns more optimistic utility to compact sub-networks and therefore selects the lower-capacity deployment, whereas deployment-aware profiling selects the higher-capacity deployment in all three regimes. Although the latter incurs a larger modeled per-round latency, its higher-capacity allocation consistently reaches a higher accuracy level. At the end of training, deployment-aware profiling improves the weighted mean node accuracy in all regimes. The resulting accuracy--latency trajectories therefore demonstrate that deployment-aware profiling affects realized decentralized learning performance through deployment selection rather than through any modification of the training algorithm. 

\begin{figure*}[t]
    \centering
    \includegraphics[width=\textwidth]
    {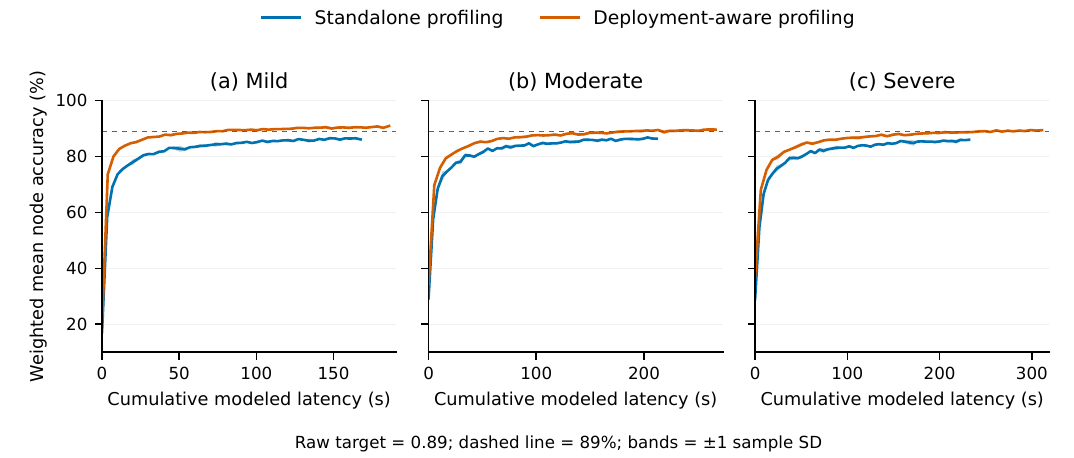}
    \caption{
    Accuracy progression of the deployments selected by standalone and deployment-aware learning-capability profiling under Mild, Moderate, and Severe resource heterogeneity.}
    \label{fig:profiling_convergence}
    \vspace{-0.5cm}
\end{figure*}

\subsection{Generalization to CIFAR-100}
\label{sec5_4}
To examine whether the proposed profiling-and-allocation framework generalizes beyond CIFAR-10, we conduct an additional experiment on CIFAR-100 while retaining the same ResNet-18 backbone and the frozen Moderate resource configuration. The training protocol is kept unchanged, including 50 communication rounds, 100 local optimization steps per node per round, a batch size of 64, and SGD with a learning
rate of $0.001$, momentum of $0.9$, and weight decay of $5\times10^{-4}$. An IID data partition is used, and all results are averaged over three
training seeds. Since the learning task changes, the deployment-aware utility profile is recomputed on CIFAR-100 and the RASA-DL deployment
frontier is reconstructed accordingly. No operating point is selected using the final test accuracy.

\begin{table}[t]
\centering
\small
\caption{Task generalization on CIFAR-100 under the frozen Moderate
resource regime. Accuracy is the final-round weighted mean node test
accuracy over three seeds, reported as mean $\pm$ sample SD. Latency is
the modeled bottleneck round latency.}
\label{tab:cifar100_generalization}
\begin{tabular}{lccc}
\hline
Method & Deployment & Latency (s) & Accuracy (\%) \\
\hline
Fixed Uniform & $[20,0,0,0]$ & 3.959 & $50.01 \pm 0.84$ \\
DFedAvgM & $[20,0,0,0]$ & 3.959 & $49.62 \pm 0.43$ \\
RASA-DL-A & $[11,4,1,4]$ & 3.959 & $58.27 \pm 0.12$ \\
RASA-DL-B & $[6,9,1,4]$ & 4.393 & $63.51 \pm 0.27$ \\
RASA-DL-C & $[5,4,4,7]$ & 6.035 & $67.16 \pm 0.49$ \\
\hline
\end{tabular}
\vspace{-0.5cm}
\end{table}

Table~\ref{tab:cifar100_generalization} summarizes the resulting accuracy--cost tradeoff. Under the frozen Moderate resource regime, the
largest globally feasible uniform deployment remains the $0.25$-width sub-network. Fixed Uniform and DFedAvgM therefore deploy the smallest sub-network on all 20 nodes and achieve $50.01\%$ and $49.62\%$ accuracy, respectively, at a modeled round latency of $3.959$~s. In contrast, RASA-DL-A assigns larger sub-networks to resource-capable nodes, with a deployment composition of $[11,4,1,4]$ over retention ratios
$[0.25,0.50,0.75,1.0]$. Despite this heterogeneous allocation, RASA-DL-A retains the same modeled round latency of $3.959$~s and achieves $58.27\%$ accuracy, corresponding to gains of $8.26$ and $8.65$ percentage points over Fixed Uniform and DFedAvgM, respectively.

The higher-capacity operating points further improve accuracy at increased modeled system cost. RASA-DL-B reaches $63.51\%$ accuracy at
$4.393$~s with deployment composition $[6,9,1,4]$, while RASA-DL-C achieves $67.16\%$ at $6.035$~s with composition $[5,4,4,7]$. These
results exhibit the same qualitative behavior observed on CIFAR-10: RASA-DL-A improves learning performance over uniform deployment without
increasing the modeled bottleneck round latency, whereas RASA-DL-B and RASA-DL-C provide progressively higher accuracy at higher modeled deployment cost.

Overall, the CIFAR-100 experiment indicates that the benefit of resource-aware heterogeneous sub-network allocation is not restricted to the original CIFAR-10 task. Under the same backbone and Moderate system realization, RASA-DL continues to identify useful heterogeneous
deployment operating points on the more challenging 100-class classification task. 

\subsection{Robustness to Statistical Heterogeneity}
\label{sec5_5}
We further evaluate whether the heterogeneous deployments identified by RASA-DL remain effective when the client data are statistically heterogeneous. To isolate the effect of data heterogeneity from system heterogeneity, we fix the Moderate resource realization, network
topology, deployment allocations, utility profile, training hyperparameters, and communication protocol throughout this study. 
Only the client data partition is changed. In particular, the CIFAR-10 training set is partitioned across the 20 nodes using Dirichlet label
skew with concentration parameters $\alpha\in\{1.0,0.5,0.2,0.1\}$, where a smaller $\alpha$ corresponds to stronger statistical heterogeneity. For each $\alpha$, a single partition realization is fixed across all methods and training seeds.
Importantly, the RASA-DL utility profile and the resulting A/B/C deployments are kept unchanged from the original deployment design.

\begin{table*}[t]
\centering
\small
\caption{Robustness under Dirichlet Non-IID data in the Moderate resource regime. Accuracy is the weighted mean node accuracy (\%) over $3$ seeds, reported as mean $\pm$ sample SD. Smaller $\alpha$ indicates stronger statistical heterogeneity. Latency is the frozen modeled bottleneck round latency.}
\label{tab:moderate-noniid-accuracy-latency}
\begin{tabular}{lccccc}
\hline
Method & Latency (s) & $\alpha=1.0$ & $\alpha=0.5$ & $\alpha=0.2$ & $\alpha=0.1$ \\
\hline
Fixed Uniform & 4.2358 & $77.32 \pm 0.60$ & $75.72 \pm 0.67$ & $50.70 \pm 2.18$ & $30.96 \pm 0.56$ \\
DFedAvgM      & 4.2358 & $78.64 \pm 0.41$ & $75.94 \pm 0.17$ & $51.23 \pm 0.31$ & $32.78 \pm 0.16$ \\
RASA-DL-A            & 4.2358 & $82.11 \pm 0.25$ & $80.14 \pm 0.65$ & $63.18 \pm 1.02$ & $49.64 \pm 1.41$ \\
RASA-DL-B            & 5.3293 & $85.85 \pm 0.07$ & $84.31 \pm 0.13$ & $73.28 \pm 0.25$ & $60.74 \pm 0.97$ \\
RASA-DL-C            & 6.1990 & $86.68 \pm 0.23$ & $84.86 \pm 0.49$ & $75.63 \pm 0.58$ & $66.45 \pm 0.59$ \\
\hline
\end{tabular}
\end{table*}

Table~\ref{tab:moderate-noniid-accuracy-latency} summarizes the final weighted mean node accuracy together with the frozen modeled bottleneck round latency. All methods experience performance degradation as the client data become more heterogeneous, particularly for $\alpha=0.2$ and $\alpha=0.1$. Nevertheless, the heterogeneous RASA-DL deployments consistently outperform both Fixed Uniform and DFedAvgM across all evaluated Dirichlet settings. A direct cost-matched comparison is provided by RASA-DL-A and Fixed Uniform-25, which have the same modeled bottleneck round latency of $4.2358$~s under the frozen Moderate system realization. RASA-DL-A improves accuracy by $4.79$, $4.42$, $12.48$, and $18.68$ percentage points for $\alpha=1.0$, $0.5$, $0.2$, and $0.1$, respectively. The improvement is positive for every training seed in all four settings, and becomes more pronounced as statistical heterogeneity increases.

These results indicate that the RASA-DL deployments remain effective under substantial statistical heterogeneity without modifying the utility profile or reallocating model capacity. In particular, when strong label skew causes the globally feasible minimum-width deployment to suffer a large accuracy loss, retaining larger sub-networks on resource-capable nodes provides a more favorable accuracy--latency tradeoff.

\section{Conclusion}\label{sec6}
This article studied resource-aware model-capacity configuration for synchronous decentralized learning over heterogeneous devices. We proposed RASA-DL, a profiling-guided sub-network allocation framework that jointly considers computation, communication, memory, and learning capability when assigning node-specific sub-networks. By formulating deployment as a bottleneck-latency minimization problem subject to memory and aggregate-utility constraints, and exploiting the monotone feasibility of the latency threshold, RASA-DL enables efficient allocation without exhaustive search.

Experiments under controlled resource heterogeneity demonstrated favorable accuracy--latency tradeoffs compared with uniform deployments and DFedAvgM. The results further validated the benefit of deployment-aware capability profiling and showed that the resulting heterogeneous deployments remain effective under Non-IID data distributions. Future work will consider time-varying resource conditions and online adaptation of model capacity, as well as joint optimization of sub-network allocation and decentralized communication strategies.

\small
\bibliographystyle{IEEEtran}
\bibliography{ref}

\end{document}